\documentclass[a4paper,twocolumn,superscriptaddress,10pt,accepted=2019-04-24]{quantumarticle} 
\pdfoutput=1
\usepackage[utf8]{inputenc}
\usepackage[english]{babel}
\usepackage[T1]{fontenc}
\usepackage[numbers,sort&compress]{natbib}
\usepackage{hyperref}

\usepackage{tikz}
\usepackage{lipsum}
\usepackage{graphicx}
\usepackage{color}
\usepackage{mathtools}
\usepackage{amsfonts}
\usepackage{amsthm}
\graphicspath{{images/}}
\usepackage{enumitem} 
\def\>{\rangle} \def\<{\langle}
\usepackage[draft,inline,nomargin]{fixme} \fxsetup{theme=color}
\renewcommand{\>}{\rangle}

\newcommand{\Div}{\ensuremath{{\sf C}^\text{div}}}
\newcommand{\pDiv}{\ensuremath{{\sf C}^\text{P}}}
\newcommand{\cpDiv}{\ensuremath{{\sf C}^\text{CP}}}
\newcommand{\LDiv}{\ensuremath{{\sf C}^\text{L}}}

\newcommand{\Ind}{\ensuremath{\overline{{\sf C}^{\rm div}}}}
\newcommand{\InftyDiv}{\ensuremath{{\sf C}^\infty}}

\newcommand{\InfDiv}{\ensuremath{{\sf C}^\text{Inf}}}

\newcommand{\DivFunc}{\delta}
\newcommand{\Jami}{Choi-Jamio\l{}kowski}

\newcommand{\one}{\openone{}}

\newtheorem{theorem}{Theorem}

\newcommand{\unam}{Universidad Nacional Aut\'onoma de M\'exico, M\'exico, D.F., M\'exico}

\newcommand{\ifunam}{Instituto de F\'{\i}sica, \unam}
\newcommand{\checo}{Faculty of Informatics, Masaryk University, Botanick\'a 68a, 60200 Brno, Czech Republic}
\newcommand{\sas}{Institute of Physics, Slovak Academy of Sciences, D\'ubravsk\'a cesta 9, Bratislava 84511, Slovakia}

\newcommand{\vienna}{Faculty of Physics, University of Vienna, 1090 Vienna, Austria}
\newcommand{\mcU}{\mathcal{U}}
\newcommand{\mcA}{\mathcal{A}}

\newcommand{\mcH}{\mathcal{H}}

\newcommand{\mcN}{\mathcal{N}}

\newcommand{\mcE}{\mathcal{E}}

\newcommand{\mcB}{\mathcal{B}}

\newcommand{\mcF}{\mathcal{F}}
\newcommand{\mcT}{\mathcal{T}}

\newcommand{\mcD}{\mathcal{D}}
\newcommand{\diag}{\text{diag}}
\newcommand{\ie}{i.e.}

\newcommand{\eref}[1]{Eq.~(\ref{#1})} 
\newcommand{\sref}[1]{sec.~\ref{#1}}
\newcommand{\fref}[1]{Fig.~\ref{#1}}

\newcommand{\tr}{\mathop{\mathrm{tr}}\nolimits}

\newcommand{\ket}[1]{{\vert #1 \rangle}}
\newcommand{\bra}[1]{{\langle #1 \vert}}

\graphicspath{{images/}}

\newcommand{\id}{\text{id}}
\begin{document}
\title{Divisibility of qubit channels and dynamical maps}
\author{David Davalos} \affiliation{\ifunam}
\email{davidphysdavalos@gmail.com}
\author{Mario Ziman} \affiliation{\sas}\affiliation{\checo}
\author{Carlos Pineda} \affiliation{\ifunam}\affiliation{\vienna}

\begin{abstract} 
The concept of divisibility of dynamical maps is used to introduce an analogous
concept for quantum channels by analyzing the \textit{simulability} of channels
by means of dynamical maps. In particular, this is addressed for
Lindblad divisible, completely positive divisible and positive divisible
dynamical maps. The corresponding L-divisible, CP-divisible and P-divisible
subsets of channels are characterized (exploiting the results by Wolf et al.
\cite{cirac}) and visualized for the case of qubit channels. We discuss the
general inclusions among divisibility sets and show several equivalences for qubit
channels. To this end we study the conditions of L-divisibility for finite dimensional channels,
especially the cases with negative eigenvalues, extending and completing the results of Ref.~\cite{Wolf2008}. Furthermore we show
that transitions between every two of the defined divisibility sets are allowed. We
explore particular examples of dynamical maps to compare these concepts.
Finally, we show that every divisible but not infinitesimal divisible qubit
channel (in positive maps) is entanglement breaking, and open the question if
something similar occurs for higher dimensions.  
 
\end{abstract} 
%
%
\maketitle
\section{Introduction} 
The advent of quantum technologies opens questions aiming for deeper understanding of the fundamental physics beyond the idealized case of isolated quantum systems. Also the well established Born-Markov approximation used to describe open quantum systems (e.g. relaxation process such as spontaneous decay) is of limited use and a more general framework of open system dynamics is demanded. Recent efforts in this area have given rise to relatively novel research subjects - non-markovianity and divisibility.

Non-markovianity is a characteristic of continuous time evolutions of quantum
systems (quantum dynamical maps), whereas divisibility refers to properties
of system's transformations (discrete quantum processes) over a fixed time
interval (quantum channels). The non-markovianity aims to capture and describe
the back-action of the system's environment on the system's future time
evolution. Such phenomena is identified as emergence
of memory effects~\cite{rivasreview,breuerreview,ourmeasure}.
On the other side, the divisibility questions the possibility of splitting
a given quantum channel into a concatenation of other quantum channels.
In this work we will investigate the relation between these two notions.

Our goal is to understand the possible forms of the dynamics
standing behind the observed quantum channels, specially in regard to their
divisibility properties which in turn determine their
markovian or non-markovian nature.
In particular, we provide characterization of the
subsets of qubit channels depending on their divisibility properties and
implementation by means of dynamical maps. An attempt to characterize
the set of channels belonging to one-parameter semigroups
induced by (time-independent) Lindblad master equations has been
already done in Ref.~\cite{Wolf2008}. However, it has drawbacks
when dealing with channels with negative eigenvalues.
Using the results of Ref.~\cite{Evans1977} and Ref.~\cite{Culver1966},
we will extend the analysis of~\cite{Wolf2008} also for channels
with negative eigenvalues.


The paper is organized as follows: In section~\ref{sec:divisibility} we give the formal definition of quantum channels and of quantum dynamical maps, and some of their properties. We discuss the meaning of divisibility for each object and discuss the known inclusions and equivalences between divisibility types.  In section~\ref{sec:qubitchan} we discuss properties and representations of qubit channels and their divisibility. We introduce a useful theorem to decide L-divisibility, which is in turn valid for any finite dimension. In section~\ref{sec:jumps} we discuss the possible transition that can be occur between divisibility types, and show two examples of dynamical maps and their transitions. Finally in Section~\ref{sec:conclusions} we summarize our results and discuss open questions.

\section{Basic definitions and divisibility} 
\label{sec:divisibility}
\subsection{Channels and divisibility classes} 

We shall study transformations of a physical system associated with a complex
Hilbert space $\mcH_d$ of dimension $d$. In particular, we consider linear maps
on bounded operators, $\mcB(\mcH_d)$, that for the finite-dimensional case coincides
with the set of trace-class operators that accommodate the subset of density
operators representing the quantum states of the system. We say a linear map
$\mcE:\mcB(\mcH_d)\to\mcB(\mcH_d)$ is \emph{positive}, if it maps positive
operators into positive operators, i.e. $X\geq 0$ implies $\mcE[X]\geq 0$.
\emph{Quantum channels} are associated with elements of the convex set ${\sf
C}$ of \emph{completely positive trace-preserving linear maps} (CPTP)
transforming density matrices into density matrices, i.e.
$\mcE:\mcB(\mcH)\to\mcB(\mcH)$ such that $\tr(\mcE[X])=\tr(X)$ for all
$X\in\mcB(\mcH)$, and all its extensions $\id_n\otimes\mcE$ are positive maps
for all $n>1$, where $\id_n$ is the identity channel on a $n$-dimensional quantum
system. 
In general a channel has the form $\mcE[X]=\sum_i K_i X K_i^{\dagger}$. The
minimum number of operators $K_i$ required in the previous expression is called
the \textit{Kraus rank} of $\mcE$.

Let us introduce two subsets of channels. First, we say a channel is
\emph{unital} if it preserves the identity operator, i.e. $\mcE[\one]= \one$.
Unital channels have a simple parametrization which will be useful for our
purposes.  Second, if $\mcE[X]=UXU^\dagger$ for some \emph{unitary operator}
$U$ (meaning $UU^\dagger=U^\dagger U=\one$), we say the channel is
\emph{unitary}. 

A quantum channel $\mcE$ is called \emph{indivisible} if it cannot be written
as a concatenation of two non-unitary channels, namely, if $\mcE =
\mcE_1\mcE_2$ implies that either $\mcE_1$, or $\mcE_2$, exclusively, is a
unitary channel. If the channel is not indivisible, it is said to be
\emph{divisible}.  We denote the set of divisible channels by ${\sf C}^{\rm
div}$ and that of indivisible channels by \Ind{}.
Following this definition, unitary channels are divisible, because for them
both (decomposing) channels $\mcE_{1,2}$ must be unitary. 
The concept of indivisible channels resembles the concept of prime numbers:
unitary channels play the role of unity (which are not indivisible/prime), i.e.
a composition of indivisible and a unitary channel results in an indivisible
channel.

We now define the set of \emph{infinitely divisible} channels (\InftyDiv) and
the set of \emph{infinitesimal divisible} channels (\InfDiv).  Infinitely
divisible channels, in some sense opposite to indivisible channels, are defined
as channels $\mcE$ for which there exist for all $n=1,2,3,\dots$ a channel
${\cal A}_n$ such that $\mcE=({\cal A}_n)^n$.  Now, consider channels $\mcE$
that may be written as products of channels close to identity, i.e. such that
for all $\epsilon>0$ there exists a finite set of channels $\varepsilon_j$ with
$||\id-\varepsilon_j||\leq \epsilon$ and $\mcE=\prod_j \varepsilon_j$. Its
closure determines the set of infinitesimal divisible channels \InfDiv.
\subsection{Quantum dynamical maps and more divisibility classes} 

The next sets of channels are going to be defined using three types of
dynamical maps.  A \emph{quantum dynamical map} is identified with a continuous
parametrized curve drawn inside the set of channels starting at the identity
channel, i.e. a one-parametric function $t\mapsto\mcE_t\in{\sf C}$ for all $t$
belonging to an interval with minimum element $0$ and satisfying the initial
condition $\mcE_0=\id$. Let $\mcE_{t,s}=\mcE_{t}^{-1}\mcE_s$ be the
linear map describing the state transformations within the time interval $[t,s]$, whenever $\mcE_{t}^{-1}$ exists.

\begin{itemize}
  \item
A given quantum dynamical map is called
\emph{CP-divisible} if for all $t<s$ the map $\mcE_{t,s}$ is a channel.
\item A given quantum dynamical map is called
  \emph{P-divisible}~\cite{breuerreview} if
$\mcE_{t,s}$ is a positive trace-preserving linear map for all $t<s$.
\item A given quantum dynamical map is called \emph{L-divisible} if
  it is induced by a time-independent Lindblad master equation~\cite{lindblad, kossa, Gorini1976}, i.e. $\mcE_t=e^{tL}$ with
\begin{equation*}
L(\rho)= i [\rho,H]
 +\sum_{\alpha,\beta} G_{\alpha \beta} 
     \left( 
         F_{\alpha}\rho F^{\dagger}_{\beta}
             -\frac{1}{2} \lbrace F^{\dagger}_{\beta} F_{\alpha},\rho \rbrace 
     \right),
\end{equation*}
where $H=H^\dagger\in\mcB(\mcH_d)$ is known as Hamiltonian,
$\{F_{\alpha}\}$ are hermitian and form an orthonormal basis of the operator
space $\mcB(\mcH_d)$, and $G_{\alpha\beta}$ constitutes a hermitian positive
semi-definite matrix.
\end{itemize}

If we allow the Lindblad generator $L$ to depend on
time, we recover the set of CP-divisible quantum dynamical maps as the
resulting dynamical maps $\mcE_t={\hat{\rm T}} e^{\int_0^t L(\tau)d\tau}$
($\hat{\rm T}$ denotes the time-ordering operator) are compositions of
infinitesimal completely-positive
maps~\cite{kossa,Gorini1976,lindblad,rivasreview}.
Notice that there is a hierarchy for quantum dynamical maps:
\emph{L-divisible} quantum dynamical maps are \emph{CP-divisible}
which in turn are \emph{P-divisible}.  

Using the introduced families of quantum dynamical maps we can now
classify quantum channels according to whether they can
be implemented by the aforementioned kinds of quantum dynamical maps.
We define subsets \LDiv{}, \cpDiv{}, and \pDiv{} of
\emph{L-divisible}, \emph{CP-divisible} and \emph{P-divisible} channels,
respectively.  In particular, we say $\mcE\in\LDiv$ if it belongs to
the closure of a L-divisible quantum dynamical map.  Let us stress that the
requirement of the existence of Lindblad generator $L$ such that $\mcE=e^{L}$
is not sufficient and closure is necessary.  For example, the evolution
governed by $L(\rho)=i[\rho, H]+\gamma [H,[H,\varrho]]$
results \cite{Ziman2005a} in the diagonalization of states
in the energy eigenbasis of $H$. Such transformation
$\mcE_{\rm diag}$ is not invertible, thus (by definition) $L=\log\mcE_{\rm
diag}$ does not exist (contains infinities).  Analogously, we say 
$\mcE\in\cpDiv$ ($\mcE\in\pDiv$) if there exists a CP-divisible
(P-divisible) dynamical map $\mcE_t$ such that $\mcE=\mcE_t$ (with arbitrary
precision) for some value of $t$ (including $t=\infty$). 

We now recall how to verify whether a channel is L-divisible 
since we will build upon the method for some of our results. 
Verifying whether $\mcE \in \LDiv$ \cite{Wolf2008} requires evaluation of
the channel's logarithms, however, the matrix logarithm is defined
only for invertible matrices and it is not unique. In fact, we need to
check if at least one of its branches has the Lindblad form.
It was shown in \cite{Evans1977} that $\mcE$ is L-divisible if and only
if there exists $L$ such that: $\exp L=\mcE$, is hermitian preserving,
trace-preserving and conditionally completely positive (ccp). Thus, we are
looking for logarithm satisfying $L(X^\dagger)=L(X)^\dagger$ (hermiticity
preserving), $L^*(\one)=0$ (trace-preserving), and 
\begin{equation}
(\one-\omega)({\rm id}\otimes
L)[\omega](\one-\omega)\geq 0
\label{eq:ccp_general}
\end{equation} (ccp condition), where
$\omega=\frac{1}{d}\sum_{j,k=1}^d\ket{j\otimes j}\bra{k\otimes k}$ is the
projector onto a maximally entangled state.
To the best of our knowledge, there is no general method to verify if a
channel is P or CP divisible, with some exceptions~\cite{cirac}.

\subsection{Relation between channel divisibility classes} 
\label{subsec:relations}
Due to the inclusion set relations between the three kinds of dynamical maps
discussed in the previous sections, we can see that 
$\LDiv\subset\cpDiv\subset\pDiv$. Similarly, from the earlier definitions, 
one can easily argue that $\InftyDiv\subset\InfDiv\subset\Div$.
By the definition of L-divisibility it is trivial to see that in
general $\LDiv{}\subset\InftyDiv$. Indeed Denisov has shown in
\cite{Denisov1989} that infinitely divisible channels can be
written as $\mcE=\mcE_0 e^{L}$, 
%
with $L$ a Lindblad  generator, and $\mcE_0$ idempotent operator such that 
$\mcE_0 L \mcE_0=\mcE_0 L$. 
Further, it was shown in Ref. \cite{cirac} that
$\mcE\in{\sf C}^{\rm inf}$ implies $\det\mcE\geq 0$ and also that $\mcE$ can be
approximated by $\prod_j e^{L_j}$, i.e. ${\sf C}^{\rm CP}={\sf C}^{\rm inf}$.
In other words, the positivity of determinant is necessary for the channel to
be (in the closure of) channels attainable by CP-divisible dynamical maps.  In
summary, we have the following relations between sets (see also
\fref{fig:setscheme}):
\begin{equation}
  \label{eq:set_relations}
\begin{array}{cccccc}
\InftyDiv{} & \subset & \InfDiv{} & \subset &\Div{} & \\ 
\rotatebox{90}{$\subseteq$} &  & \rotatebox{90}{=} &  &   \\ 
\LDiv{} & \subset & \cpDiv{} & \subset & \pDiv{} 
\end{array} .
\end{equation}
The relation between \pDiv{}
and \Div{} is unknown, although it is clear that $\Div{} \subset \pDiv{}$ is not possible since channels in \pDiv{} are not necessarily divisible in CP maps. The intersection of \pDiv{} and \Div{} is not empty since $\cpDiv{}\subseteq \Div{}$ and $\cpDiv{} \subseteq \pDiv{}$. Later on we will investigate if $\pDiv{}\subseteq \Div{}$ or not.

\begin{figure} 
\centering
\includegraphics{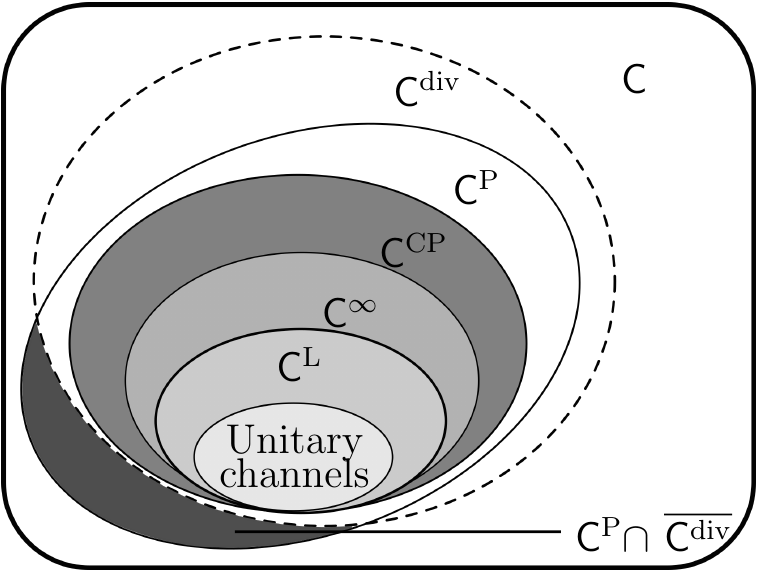}
\caption{
Scheme illustrating the different sets of quantum channels for a given dimension, 
discussed in \sref{sec:divisibility}. In particular, the inclusion relations
presented in \eref{eq:set_relations} are depicted.
\label{fig:setscheme}}
\end{figure} 

\section{Qubit channels} 
\label{sec:qubitchan}
\subsection{Representations} 
\label{subsec:representations}
Using the Pauli basis $\frac{1}{\sqrt{2}}\lbrace
\mathbf{1},\sigma_x,\sigma_y,\sigma_z\rbrace$, and the standard Hilbert-Schmidt
inner product, the real representation for qubit channels is given by
\cite{ruskai,zimansbook}:
\begin{equation}
\hat \mcE=\left(\begin{array}{cc}
1 & \vec 0^{T} \\ 
\vec t & \Delta
\end{array} \right).
\end{equation}
This describes the action of the channel in the Bloch sphere picture in which
the points $\vec{r}$ are identified with density operators
$\varrho_{\vec{r}}=\frac{1}{2}(I+\vec{r}\cdot\vec{\sigma})$. We will write
$\mcE=(\Delta,\vec{t})$ meaning that $\mcE(\rho_{\vec{r}})
=\rho_{\Delta\vec{r}+\vec{t}}$.

In order to study qubit channels with simpler expressions, we will consider a decomposition in unitaries such that
\begin{equation}
\mcE=\mcU_1 \mcD \mcU_2.
\label{eq:EUDU}
\end{equation}
This can be performed by decomposing $\Delta$ in rotation matrices,
\ie{} $\Delta=R_1 D R_2$, where $D={\rm diag}(\lambda_1,\lambda_2,\lambda_3)$
is diagonal and the rotations $R_{1,2}\in\text{SO}(3) $ (of
the Bloch sphere) correspond to the unitary channels $\mcU_{1,2}$. This
decomposition should not be confused with the singular value
decomposition. The latter allows decompositions that include, say, total
reflections. Such operations do not correspond to unitaries over a qubit,
in fact they are not CPTP. 
Therefore the channel $\mcD$, in the Pauli basis, is given by
\begin{equation}
\hat \mcD=\left( \begin{array}{cc}
1 & \vec 0^{T} \\ 
\vec{\tau} & D
\end{array} \right)\,,
\label{eq:orthogonalform}
\end{equation}
where $\Delta=R_1 DR_2$ and $\vec{\tau}=R_1^T\vec{t}$. The latter describes the
shift of the center of the Bloch sphere under the action of $\mcD$.
The parameters $\vec{\lambda}$ determine the length of semi-axes
of the Bloch ellipsoid, being the deformation of Bloch sphere under
the action of $\mcE$. From now we will call the form $\mcD$, \textit{special
orthogonal normal form}.


We shall develop a geometric intuition in the space determined by the possible
values of these three parameters. For an arbitrary channel, complete positivity
implies that the possible set of lambdas lives inside the tetrahedron with
corners $(1,1,1)$, $(1,-1,-1)$, $(-1,1,-1)$ and $(-1,-1,1)$, see
\fref{fig:tetra}.  For unital channels, all points in the tetrahedron are
allowed, but for non-unital channels more restrictive conditions arise. In
\fref{fig:cutnonunital2} we present a visualization of the permitted values of
$\vec \lambda$ for a particular nontrivial value of $\vec \tau$, and in
\cite{geometry2017} the steps to study the general case from an algebraic point of
view are presented.  For the unital case, the corner $\vec \lambda = (1,1,1)$
corresponds to the identity channel, $\vec \lambda = (1,-1,-1)$ to $\sigma_x$
$\vec \lambda = (-1,1,-1)$ to $\sigma_y$  and $\vec \lambda = (-1,-1,1)$ to
$\sigma_z$ (Kraus rank 1 operations). Points in the edges correspond to Kraus
rank 2 operations, points in the faces to Kraus rank 3 operations and in the
interior of the tetrahedron to Kraus rank 4 operations.

In addition to this decomposition, following the definition of divisibility, concatenation with unitaries of a given quantum channel do not change the divisibility character of the latter. Thus, orthogonal normal forms are useful to study divisibility since, following also the properties of \pDiv{} and \cpDiv{} introduced in \sref{subsec:representations}, immediately one has:
\begin{theorem}[\textbf{Divisibility of
special orthogonal normal forms}]
Let $\mcE$ a qubit quantum channel and $\mcD$ its special orthogonal normal
form, $\mcE$ belongs to ${\sf C}^X$ if and only if $\mcD$ does, where $X=\lbrace
\text{``Div'', ``P'', ``CP''} \rbrace$. 
\label{thm:divisibility_using_orthogonal_form}
\end{theorem}

\begin{figure} 
\centering
\includegraphics{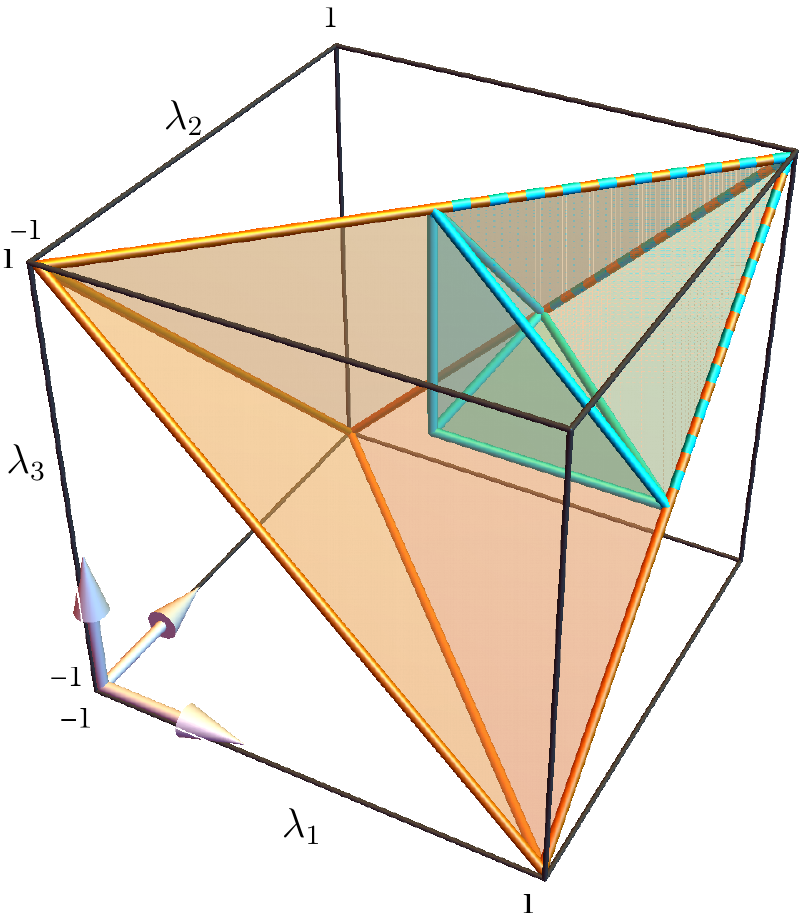}
\caption{Tetrahedron of Pauli channels. The corners correspond to unitary Pauli
operations ($\openone$, $\sigma_{x,y,z}$) while the rest can be written as
convex combinations of them.
The bipyramid in blue corresponds to channels with $\lambda_i>0 \forall i$,
\ie{} channels of the positive octant belonging to \pDiv{}. The whole set
\pDiv{} includes three other bipyramids corresponding to the other vertexes of
tetrahedron, \ie{} \pDiv{} enjoys the symmetries of the tetrahedron, see \eref{eq:pdiv_qubits}.  The faces of the
bipyramids matching the corners of the tetrahedron are subsets of the faces
of the tetrahedron, \ie{} contain Kraus rank three channels. Such channels are
then \pDiv but also \Ind{}, showing that the intersection shown in
\fref{fig:setscheme} is not empty.
\label{fig:tetra}
}
\end{figure} 

There is another another parametrization for qubit channels called
\textit{Lorentz normal decomposition}~\cite{Verstraete2001,Verstraete2002} which
is specially useful to characterize infinitesimal divisibility \InfDiv{}, 
and geometric aspects of entanglement~\cite{Myrheim}. 
This decomposition is derived from the theorem 3 of Ref.~\cite{Verstraete2001},
which essentially states that for a qubit state $\rho=\frac{1}{4}\sum_{i,j=0}^3
R_{ij}\sigma_i \otimes \sigma_j$ the matrix $R$ can be decomposed as $R=L_1
\Sigma L^\text{T}_2$. Here $L_{1,2}$ are proper orthochronous Lorentz transformations
and $\Sigma$ is either $\Sigma=\diag \left( s_0, s_1, s_2, s_3 \right)$ with
$s_0 \ge s_1 \ge s_2 \ge |s_3|$, or 
\begin{equation}
\Sigma=\left(
\begin{array}{cccc}
 a & 0 & 0 & b \\
 0 & d & 0 & 0 \\
 0 & 0 & -d & 0 \\
 c & 0 & 0 & -b+c+a 
\end{array}
\right).
\label{eq:state_normal_form_singular}
\end{equation}
In theorem 8 of Ref.~\cite{Verstraete2002} the authors make a similar claim,
exploiting the \Jami{} isomorphism. They forced
$b=0$ in order to have normal forms proportional to trace-preserving operations
in the case of Kraus rank deficient ones, see
\eref{eq:state_normal_form_singular}. The latter is equivalent to saying that the
decomposition of \Jami{} states leads to states that are proportional to
\Jami{} states. We didn't find a good argument to justify such an assumption and
found a counterexample (see appendix~\ref{sec:appendix}). Thus, we propose a
restricted version of their theorem:
\begin{theorem}[\textbf{Restricted Lorentz normal form for qubit quantum channels}]
For any full Kraus rank qubit channel $\mcE$ there exists rank-one completely
positive maps $\mcT_1,\mcT_2$ such that $\mcT=\mcT_1\mcE\mcT_2$ is proportional
to 
\begin{equation}
\left( \begin{array}{cc}
1 & \vec 0^\text{T} \\ 
\vec 0 & \Lambda
\end{array}  \right),
\label{eq:form_fullKraus_Lorentz}
\end{equation}
where
$\Lambda={\rm
        diag}(s_1,s_2,s_3)$ with $ 1\geq s_1 \geq s_2 \geq |s_3|$.
\label{thm:Lorentz}
\end{theorem}

The channel $\mcT$ is called the Lorentz normal form of the channel $\mcE$. For unital qubit channels $D$ coincides with $\Lambda$, thus in such case the form of \eqref{eq:form_fullKraus_Lorentz} holds for any Kraus rank.

\subsection{Divisibility} 

In this subsection we will recall the criteria to decide if a qubit channel
belongs to \Div{}, \pDiv{} and \cpDiv{} following \cite{cirac}. We shall
start with some general statements, and then focus on different types of channels (unital, diagonal non-unital and general ones).  We will also discuss in detail the characterization of \LDiv{}, which entails a higher complexity.

It was shown (\cite{cirac}, Theorem 11) that full Kraus rank
channels are divisible (\Div{}). This simply means that all points in the interior of
the set of channels correspond to divisible channels. 
Moreover, according to Theorem 23 of the same reference, qubit channels are
indivisible if and only if they have Kraus rank three and diagonal Lorentz
normal form.
Notice that since we dispute the theorem upon which such statement is based, 
the classification might be inaccurate, see the appendix. 
It follows from the definition that for qubit channels $\mcE$ is divisible if
and only if $\mcD$ is divisible. To test if $\mcD$ is divisible, we check that
all eigenvalues of its Choi matrix are different from zero.
 
A non-negative determinant of $\mcE$ is a necessary condition  for a general
channel to belong to \pDiv{} (\cite{cirac}, Proposition 15). For qubits, this is also sufficient (\cite{cirac}, Theorem 25), and
given that $\det{\mcD}=\det{\mcE}$, the condition for qubit channels
simply reads
\begin{equation}
\det \mcE=\lambda_1 \lambda_2 \lambda_3 \geq 0.
\label{eq:pdiv_qubits}
\end{equation}
However, to our knowledge, a simple condition for arbitrary
dimension is yet unknown. 

With respect to testing for CP-divisibility we restrict the discussion to qubit
channels. To characterize CP-divisible channels it is useful to consider the
Lorentz normal form for channels. 
A full Kraus rank qubit channel $\mcE$ belongs
to \cpDiv{} if and only if it has diagonal Lorentz normal form with
\begin{equation}
s_{\min}^2\geq s_1 s_2 s_3> 0
\label{eq:cpfullkraus_qubits}
\end{equation}
where $s_{\min}$ is the smallest of $s_1$, $s_2$ and $s_3$, see
theorem~\ref{thm:Lorentz} and \cite{Wolf2008}. 
For Kraus deficient channels the pertinent theorems are based on
Kraus deficient Lorentz normal forms that according to our appendix should
be reviewed.

Deciding L-divisibility, as mentioned above, is equivalent to proving the
existence of a hermiticity preserving generator which additionally fulfills the ccp
condition.

To prove the former we recall that every hermiticity preserving operator has
a real matrix representation when choosing a hermitian basis. Since quantum
channels preserve hermiticity, the problem is reduced on finding a real
logarithm $\log\hat\mcE$ given a real matrix $\hat \mcE$. This problem was
already solved by Culver~\cite{Culver1966} who characterized
completely the existence of real logarithms of real matrices. For diagonalizable matrices the results can
be summarized as follows:

\begin{theorem}[\textbf{Existence of hermiticity preserving generator}]
A non-singular matrix with real entries $\hat \mcE$ has a real generator
(\ie{} a $\log \hat \mcE$ has real entries) if and only if
the spectrum fulfills the following conditions:
\begin{enumerate}[label=(\roman*)]
\item negative eigenvalues are even-fold degenerate;
\item complex eigenvalues come in complex conjugate pairs.
\end{enumerate}
\label{thm:culver}
\end{theorem} 

Let us examine this theorem for the particular case of qubits.
In this case this theorem means that real logarithm(s)
of $\hat \mcE$ exist if and only if $\mcE$ has either only positive eigenvalues, one
positive and two complex, or one positive and two equal non-positive
eigenvalues, apart from the trivial eigenvalue equal to one. Notice that
quantum channels with complex eigenvalues will
fulfill the last condition immediately since they preserve hermiticity.

We now continue discussing the multiplicity of the solutions of 
$\log \hat \mcE$, as finding an appropriate parametrization   
is essential to test for the ccp condition, see~\eref{eq:ccp_general}.
If $\hat \mcE$ has positive degenerate, negative, or complex
eigenvalues,
its real logarithms are not unique, and are spanned by \textit{real
logarithm
branches}~\cite{Culver1966}. In case of having negative eigenvalues,
it turns out that real logarithms always have a continuous parametrization, in addition to real branches due to the freedom of the Jordan normal
form transformation matrices. Given a real representation of $\mcE$, \ie{}
$\hat \mcE$, the Jordan form is given by $\hat \mcE=w J w^{-1}=\tilde w J
\tilde w^{-1}$, where $w = \tilde w K$ with $K$ 
belonging to a continuum of matrices that commutes with $J$~\cite{Culver1966}.
In the case of diagonalizable matrices, if there are no degeneracies, $K$
commutes with $\log(J)$.

Finally let
us note that if a channel belongs to \LDiv{}, unitary conjugations can bring it
to $\InfDiv{}\setminus\LDiv{}$ and vice versa.

\subsection{Unital channels} 
\label{subsubsec:unital}
We shall start our study of unital qubit channels, by considering Pauli
channels, defined as convex combinations of the unitaries $\sigma_i$:
\begin{equation}	
\mcE_\text{Pauli}[\rho]=\sum_{i=0}^3 p_i \sigma_i \rho \sigma_i,
\end{equation}
where $\sigma_0=\one$ and $p_i\geq 0$ with $\sum_i p_i=1$.
%
%
The special orthogonal normal form of a Pauli channel [see Eqs.~(\ref{eq:EUDU})
and (\ref{eq:orthogonalform})] has 
$\mcU_1=\mcU_2=\id{}$ and $\vec \tau = \vec0$. 
%
%
%
Thus, Pauli channels are fully characterized only by $\vec \lambda$.
Notice that every unital qubit channel can be written as 
\begin{equation}
\mcE_\text{unital}=\mcU_1 \mcE_{\text{Pauli}}\mcU_2.
\label{eq:pauli_unital}
\end{equation} 
This implies that arbitrary unital qubit channels can be expressed
as convex combinations of (at most) four unitary channels. 

Following theorem~\ref{thm:divisibility_using_orthogonal_form} it is straightforward to note that by characterizing \Div{}, \pDiv{} and \cpDiv{} of Pauli channels, the same conclusions hold for general unital qubit channels having the same $\vec \lambda$.
Additionally we can have a one-to-one geometrical view of the divisibility sets for Pauli channels given they have a one-to-one correspondence to the tetrahedron shown in \fref{fig:tetra}, defined by the inequalities \begin{align}
1+\lambda_i -\lambda_j -\lambda_k &\geq 0 \\
1+\lambda_1 +\lambda_2 +\lambda_3 &\geq 0
\label{eq:cp_inequalities}
\end{align}
with $i$, $j$ and $k$ all different~\cite{Ziman2005}.

\subsubsection{P-divisibility}

Let us discuss the divisibility properties of Pauli channels.
Divisibility in CPTP (\Div{}) is guaranteed for full Kraus rank channels, \ie{} for the interior channels of the tetrahedron. For Pauli channels this is equivalent to taking only the inequality of equations (\ref{eq:cp_inequalities}).
The characterization of
\pDiv{} can be done directly using \eref{eq:pdiv_qubits}, as it
depends only on $\vec \lambda$. This set is the intersection
of the tetrahedron with the octants 
where the product  of all $\lambda$s is positive. In fact, it consists of 
four triangular bipyramids starting in each vertex of the tetrahedron and
meeting in its center, see~\fref{fig:tetra}.
Let us study the intersection of this set with the set of unital
entanglement breaking (EB)
channels~\cite{Ziman2005}, forming an octahedron (being the intersection of the
tetrahedron with its space inversion, see~\fref{fig:bypy}). It is defined by the inequalities 
\begin{align}
\lambda_1 + \lambda_2 + \lambda_3 &\leq 1\nonumber\\
\lambda_i - \lambda_j - \lambda_k &\leq 1,
\label{eq:eb_inequalities}
\end{align} 
with $i$, $j$ and $k$ all different~\cite{Ziman2005}, together with \eref{eq:cp_inequalities}.
It follows that unital qubit channels that are not achieved by P-divisible
dynamical maps are necessarily entanglement breaking (see~\fref{fig:bypy} and
\fref{fig:cut1}). In fact this holds for general qubit channels, see
section~\ref{subsec:generalqubitchannels}.
\begin{figure} 
\centering
\includegraphics{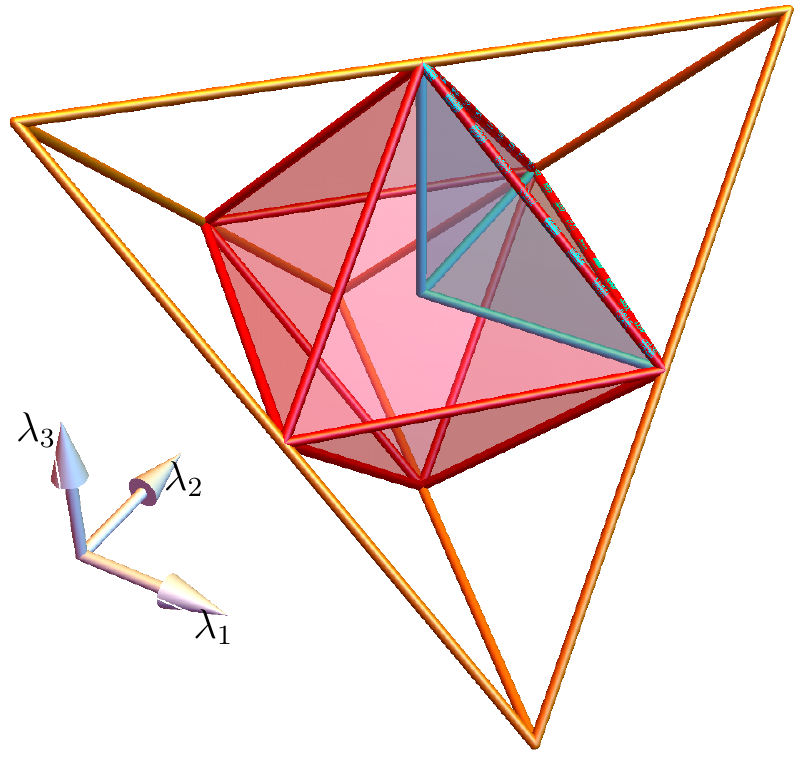}
\caption{Tetrahedron of Pauli channels with the octahedron of entanglement
breaking channels shown in red, see \eref{eq:eb_inequalities}. The blue pyramid inside the octahedron is the
intersection of the bipyramid shown in \fref{fig:tetra}, with the octahedron.
The complement of the intersections of the four bipyramids forms the set of
divisible but not infinitesimal divisible channels in PTP. Thus, a central feature of
the figure is that the set \Div{}$\setminus$\pDiv{} is always
entanglement breaking, but the converse is not true. \label{fig:bypy}}
\end{figure} 

\subsubsection{CP-divisibility}
The subset of CP-divisible Pauli channels, following \eref{eq:cpfullkraus_qubits} and theorem~\ref{thm:Lorentz},
is determined by the inequalities
\begin{equation}
  0<\lambda_1 \lambda_2 \lambda_3\leq \lambda_{\min}^2\,.
\label{eq:cpdivqubit}
\end{equation}
They determine a body
that is symmetric with respect to permutation of Pauli unitary channels (i.e.
in $\lambda_j$), hence, the set of \cpDiv{} of Pauli channels possesses the
symmetries of the tetrahedron. The set \cpDiv{}$\setminus$\LDiv{} is plotted
in~\fref{fig:cp}. 
Notice that this set coincides with the set of
unistochastic qubit channels, see Ref.~\cite{unistochastic}.

\subsubsection{L-divisibility}
Let us now derive the conditions for L-divisibility of Pauli channels with
positive eigenvalues $\lambda_1,\lambda_2,\lambda_3$ ($\lambda_0=1$). 
The logarithm of $\mcD$, induced by the principal logarithm of its
eigenvalues, is thus
\begin{equation}
L=K{\rm diag}(0,\log\lambda_1,\log\lambda_2,\log\lambda_3)K^{-1}\,,
\label{eq:logL_positive}
\end{equation}
which is real (hermiticity preserving). In case of no-degeneration
the dependency on $K$ vanishes and $L$ is unique. In such
case the ccp conditions
$\log \lambda_j-\log \lambda_k-\log \lambda_l\geq 0$ imply
\begin{align}
  \lambda_j\lambda_k\leq \lambda_l
\label{eq:hermpres}
\end{align}
for all combinations of mutually different $j,k,l$. This set (channels
belonging to \LDiv{} with positive eigenvalues)  forms a three dimensional
manifold, see \fref{fig:markov}. 

In case of degeneration, let us label the eigenvalues $\eta$, $\lambda$ and $\lambda$. In this case, the real solution for $L$ is not unique
and is parametrized by real branches in the degenerate 
subspace and by the continuous parameters of $K$~\cite{Culver1966}. 
Let us study the principal branch with $K=\one$. \eref{eq:hermpres} is
then reduced to 
\begin{equation}
\lambda^2\leq \eta \leq 1\;.
\label{eq:ccp_degenerate}
\end{equation}
Therefore, if these inequalities are fulfilled, the generator has Lindblad
form. If not, then \textit{a priori} other branches can fulfill the ccp condition and
consequently have a Lindblad form. Thus, \eref{eq:ccp_degenerate} provides
a sufficient condition for the channel to be in \LDiv{}. We will
see it is also necessary. 

Indeed, the complete positivity condition requires $\eta,\lambda\leq 1$, thus,
it remains to verify only the condition $\lambda^2\leq \eta$. It holds
for the case $\lambda\leq\eta$. If $\eta\leq\lambda$, then this condition
coincides with the CP-divisibility condition from \eref{eq:cpdivqubit}.
Since \LDiv{} implies \cpDiv{} the proof is completed. In conclusion,
the condition in \eref{eq:hermpres} is a necessary and sufficient
condition for a given Pauli channel with positive eigenvalues to belong to \LDiv{}.

Let us stress that the obtained subset of L-divisible channels
does not possess the tetrahedron symmetries. In fact, composing
$\mcD$ with a $\sigma_z$ rotation $\mcU_z={\rm diag}(1,-1,-1,1)$
results in the Pauli channel
$\mcD^\prime={\rm diag}(1,-\lambda_1,-\lambda_2,\lambda_3)$.
Clearly, if $\lambda_j$ are positive ($\mcD$ is L-divisible),
then $\mcD^\prime$ has non-positive eigenvalues. Moreover, if all $\lambda_j$
are different, then $\mcD^\prime$ does not have any
real logarithm, therefore, it cannot be L-divisible.
In conclusion, the set of L-divisible unital qubit channels is
not symmetric with respect to tetrahedron symmetries.

In what follows we will investigate the case of non-positive eigenvalues.
Theorem 3 implies that that eigenvalues have the form (modulo permutations)
$\eta,-\lambda,-\lambda$, where $\eta,\lambda\geq 0$. The corresponding
Pauli channels are $\mcD_x={\rm diag}(1,\eta,-\lambda,-\lambda)$,
$\mcD_y={\rm diag}(1,-\lambda,\eta,-\lambda)$,
$\mcD_z={\rm diag}(1,-\lambda,-\lambda,\eta)$, thus
forming three two-dimensional regions inside the tetrahedron.
Take, for instance, $\mcD_z$ that specifies a plane (inside the tetrahedron)
containing $I$, $\sigma_z$ and completely depolarizing channel
$\mcN={\rm diag}(1,0,0,0)$. The real logarithms for this case are given by
\begin{equation}
L=K\left(
\begin{array}{cccc}
0 & 0 & 0 & 0 \\ 
0 & \log( \lambda ) & (2 k +1) \pi & 0 \\ 
0 & -(2k +1)\pi & \log( \lambda ) & 0 \\ 
0 & 0 & 0 & \log( \eta )
\end{array} 
\right) K^{-1},
\label{eq:Lfornegative}
\end{equation}
where $k\in \mathbb{Z}$ and $K$, as mentioned above, belongs to a continuum of
matrices that commute with $\mcD_z$. Note that $L$ is always non-diagonal. For
this case (similarly for $\mcD_x$ and $\mcD_y$) the ccp condition reduces
again to conditions specified in \eref{eq:ccp_degenerate}. Using the same
arguments one arrives to a more general conclusion: \eref{eq:hermpres}
provides necessary and sufficient conditions for $L-divisibility$ of
a given Pauli channel, and if it is the case, the principal branch
with $K=\one$ has Lindblad form. The set of L-divisible
Pauli channels is illustrated in \fref{fig:markov}.

  In order to decide L-divisibility of general unital channels it remains to analyze the case of complex eigenvalues. The logarithms are parametrized as follows
\begin{equation}
L_{k,K}=w
K\log\left(J \right)_kK^{-1}w^{-1},
\label{eq:general_parametrization_of_L}
\end{equation} 
where $J$ is the \textit{real Jordan form} of the (qubit)
channel~\cite{Culver1966}:
\begin{equation}
J=\diag\left(1,c\right)\oplus\left( \begin{array}{cc}
a & -b \\ 
b & a
\end{array}  \right)
\label{eq:real_jordan_form}
\end{equation}
with $a\pm i b$ being the complex eigenvalues and $c>0$.
Let us note that $K\log\left(J \right)_kK^{-1}$ is reduced to equations
(\ref{eq:logL_positive}) and (\ref{eq:Lfornegative}) in the case of real
eigenvalues. In general, the generator is (up to diagonalization): 
\begin{multline}
K\log \left(J \right)_kK^{-1}=K\diag\left(0,\log(c)\right)
\oplus\nonumber\\ 
\left( \begin{array}{cc}
\log(\vert z \vert) & \arg(z) +2 \pi k \\ 
-\arg(z) -2 \pi k & \log(\vert z \vert)
\end{array} \right)K^{-1}.
\end{multline}
with $z=a+ib$.
The non-diagonal block of the logarithm has the same
structure as the real Jordan form of the channel, so $K$ also commutes with
$\log(J)_k$, leading to a countable parametric space of hermitian preserving
generators. In fact, generators of diagonalizable channels have continuous
parametrizations if and only if they have degenerate eigenvalues; the
non-diagonalizable case can be found elsewhere~\cite{Culver1966}. 
Since we are dealing with a diagonalization, the ccp condition can be very
complicated and depends in general on $k$, see ~\eref{eq:ccp_general}. But for
the complex case we can simplify the condition for qubit channels which have
exactly the form presented in ~\eref{eq:real_jordan_form}, say $\hat \mcE_\text{complex}$, \ie{} $w=\one$. In
such case the ccp condition is reduced to
\begin{equation}
a^2+b^2\leq c \leq 1.
\label{eq:ccp_complex}
\end{equation}
Note that it does not depend on $k$ and the second inequality is always
fulfilled for CPTP channels.


We can present the conditions for L-divisibility for the case of complex
eigenvalues.
The orthogonal normal form of $\hat \mcE_\text{complex}$ is 
$\hat\mcD=\diag \left(1,\eta,\lambda,\lambda \right)$ with $\eta=c$ and
$\lambda=\text{sign}(ab) \sqrt{a^2+b^2}$. The
ccp condition for degenerate eigenvalues, see~\eref{eq:ccp_degenerate}, is
reduced to~\eref{eq:ccp_complex} for this case. Therefore, the L-divisible
channels with form $\hat\mcE_\text{complex}$ are also L-divisible, up to
unitaries. This also applies for channels arising from changing positions of
the $1\times 1$ block containing $c$ and the $2\times 2$ block containing $a$
and $b$ in $\hat\mcE_\text{complex}$, with orthogonal normal forms $\diag
\left(1,\lambda,\eta,\lambda \right)$ and $\diag \left(1,\lambda,\lambda,\eta
\right)$. The set containing them is shown
in~\fref{fig:tetra_complex_eigenvalues}.
 
%
\begin{figure} 
\centering
\includegraphics{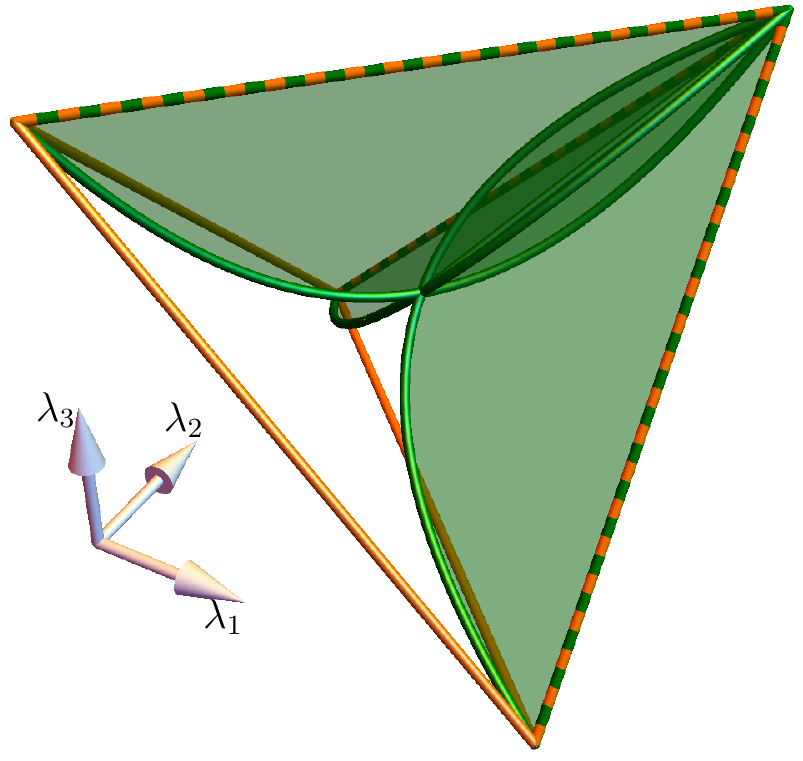}
\caption{
Tetrahedron of Pauli channels, with qubit unital L-divisible
channels of the form $\hat\mcE_\text{complex}$ (see main text). Note that the set
does not have the symmetries of the tetrahedron.
\label{fig:tetra_complex_eigenvalues}}
\end{figure} 
%
%



\subsubsection{Divisibility relations}

Consider a Pauli channel with
$0<\lambda_{\min}=\lambda_1\leq\lambda_2\leq\lambda_3<1$, thus, the condition
$\lambda_1\lambda_2\leq \lambda_3$ trivially holds. Since
$\lambda_1\lambda_2\leq \lambda_1\lambda_3\leq \lambda_2\lambda_3\leq\lambda_2$,
it follows that $\lambda_1\lambda_3\leq\lambda_2$, thus, two (out of three)
L-divisibility conditions hold always for Pauli channels with
positive eigenvalues. Moreover, one may observe that CP-divisibility
condition \eref{eq:cpdivqubit} reduces to one
of L-divisibility conditions $\lambda_2\lambda_3\leq\lambda_1$.
In conclusion, the conditions of CP-divisibility and L-divisibility
for Pauli channels with positive eigenvalues coincide, thus,
in this case \cpDiv{} implies \LDiv{}.

Concatenating (positive-eigenvalues) Pauli channels with
$\mcD_{x,y,z}$ one can generate the whole set of \cpDiv{} Pauli channels.
Using the identity $\cpDiv{}=\InfDiv{}$ and considering
\eref{eq:pauli_unital}we can formulate the following theorem:
\begin{theorem}[Infinitesimal divisible unital channels]
Let $\mcE^{\text{CP}}_{\text{unital}}$ be an arbitrary infinitesimal divisible unital qubit channel. There exists at least one L-divisible Pauli channel $\tilde \mcE$, and two unitary conjugations $\mcU_1$ and $\mcU_2$, such that
$$\mcE^\text{CP}_{\text{unital}}=\mcU_1  \tilde \mcE \mcU_2\,.$$
Notice that if $\mcE^\text{CP}_{\text{unital}}$ is invertible, $\tilde \mcE=e^L$.
\end{theorem}

\begin{figure} 
\centering
\includegraphics{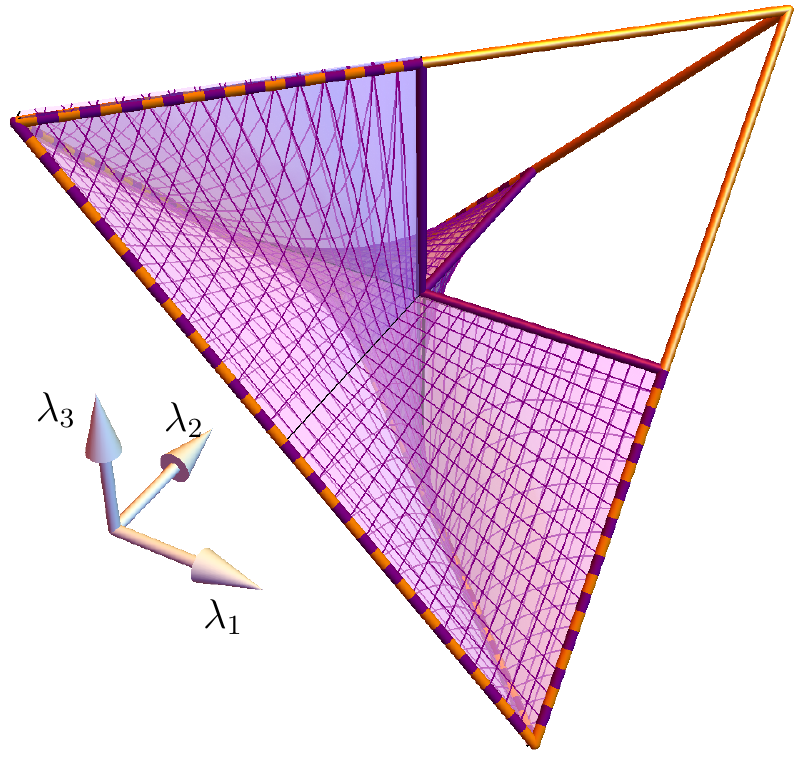}
\caption{Tetrahedron of Pauli channels with part of the set of CP-divisible, 
see \eref{eq:cpdivqubit}, but
not L-divisible channels (\cpDiv{}$\setminus$\LDiv{}) shown in purple. The
whole set \cpDiv{} 
is obtained applying the symmetry transformations of the tetrahedron to the 
purple volume. 
\label{fig:cp}}
\end{figure} 

Let us continue with another equivalence relation holding for Pauli channels. Regarding infinitely divisibility channels, we know that, in general, $\text{\LDiv{}}\subset\text{\InftyDiv{}}$, however, for Pauli channels the corresponding subsets coincide.

\begin{theorem}[Infinitely divisible Pauli channels]
The set of L-divisible Pauli channels is
equivalent to the set of infinitely divisible Pauli channels.
\label{thm:pauli_infinity}
\end{theorem}
\begin{proof}
A channel
is infinitely divisible if and only if it can be written as $\mcE_0 e^L$, where $\mcE_0$
is an idempotent channel satisfying $\mcE_0 L \mcE_0=\mcE_0 L$ and $L$ has
Lindblad form~\cite{Denisov1989}. 
The only idempotent qubit channels are contractions of the Bloch
sphere into single points, diagonalization channels $\mcE_\diag$ transforming
Bloch sphere into a line connecting a pair of basis states, and the identity
channel. 
Among the single-point contractions, the only one that is a Pauli channel is the contraction of the Bloch sphere into the complete mixture. In particular, $\mcE = \mcN e^L=\mcN$ for all $L$. The channel $\mcN$
belongs to the closure of \LDiv{}, because a sequence of channels $e^{L_n}$
with $\hat{L}_n=\diag \left(0,-n,-n,-n \right)$ converges to $\hat\mcN$ in the
limit $n\to\infty$. 
For the case of $\mcE_0$ being the identity channel we have $\mcE=e^L$, thus,
trivially such infinitely divisible channel $\mcE$ is in \LDiv{} too. 
It
remains to analyze the case of diagonalization channels. First, let us
note that the matrix of $e^{\hat{L}}$ is necessarily of full rank,
since ${\rm det}\hat{\mcE}\neq 0$. It follows that the matrix
$\hat{\mcE}=\hat\mcE_{\diag} e^{\hat{L}}$ has rank two as $\hat\mcE_{\diag}$
is a rank two matrix, thus, it takes one of the following forms
$\hat{\mcE}_{x}^\lambda=\diag \left(1,\lambda,0,0\right)$,
$\hat{\mcE}_{y}^\lambda=\diag \left(1,0,\lambda,0\right)$,
$\hat{\mcE}_{z}^\lambda=\diag \left(1,0,0,\lambda\right)$.
The infinitely divisibility implies $\lambda>0$ in order to keep
the roots of $\lambda$ real. In what follows we will show that
$\hat{\mcE}_z$ belongs to (the closure of) \LDiv{}. Let us
define the channels $\hat\mcE_z^{\lambda,
\epsilon}=\diag\left(1,\epsilon,\epsilon,\lambda \right)$ with $\epsilon>0$.
The complete positivity and ccp conditions
translate into the inequalities $\epsilon\leq \frac{1+\lambda}{2}$
and $\epsilon^2\leq \lambda$, respectively; therefore one can always find an
$\epsilon>0$ such that $\hat \mcE_z^{\lambda, \epsilon}$ is a L-divisible
channel. 
If we choose $\epsilon=\sqrt{\lambda}/n$ with $n\in
\mathbb{Z}^+$, the
channels $\hat
\mcE_{z,n}=\diag\left(1,\sqrt{\lambda}/n,\sqrt{\lambda}/n,\lambda \right)$ form
a sequence of L-divisible channels converging to $\hat\mcE_z^\lambda$ when
$n\to\infty$. The analogous reasoning implies that $\hat{\mcE}_x^\lambda,
\hat{\mcE}_y^\lambda\in\LDiv{}$ too. Let us note that one parameter
family $\mcE_z$ are convex combinations of the complete
diagonalization channel
$\hat{\mcE}_z^1=\diag \left(1,0,0,1\right)$ 
and the complete mixture contraction $\hat \mcN$.
This completes the proof.
\end{proof}

Finally, let us remark that using theorem 23 of Ref.~\cite{cirac} we
conclude that the intersection $\pDiv{}\cap \Ind{}$ depicted
in~\fref{fig:setscheme} is not empty. To show this, notice that applying the
mentioned theorem to Pauli channels we get that the faces of the tetrahedron
are indivisible in CPTP channels. However, there are channels with positive
determinant inside the faces, for example $\diag
\left(1,\frac{4}{5},\frac{4}{5},\frac{3}{5}\right)$. Therefore we conclude that
up to unitaries, $\pDiv{}\cap \Ind{}$ correspond to the union of the four faces
faces of the tetrahedron minus the faces of the octahedron that intersects with
the faces of the tetrahedron, see \fref{fig:bypy}. We have to remove such
intersection since it corresponds to channels with negative determinant, \ie{}
not in \pDiv{}.

To get a detailed picture of the position and inclusions of the divisibility
sets, we illustrate in \fref{fig:cut1} two slices of the tetrahedron where
different types of divisibility are visualized. Notice the non-convexity of the
considered divisibility sets.

\begin{figure} 
\centering
\includegraphics{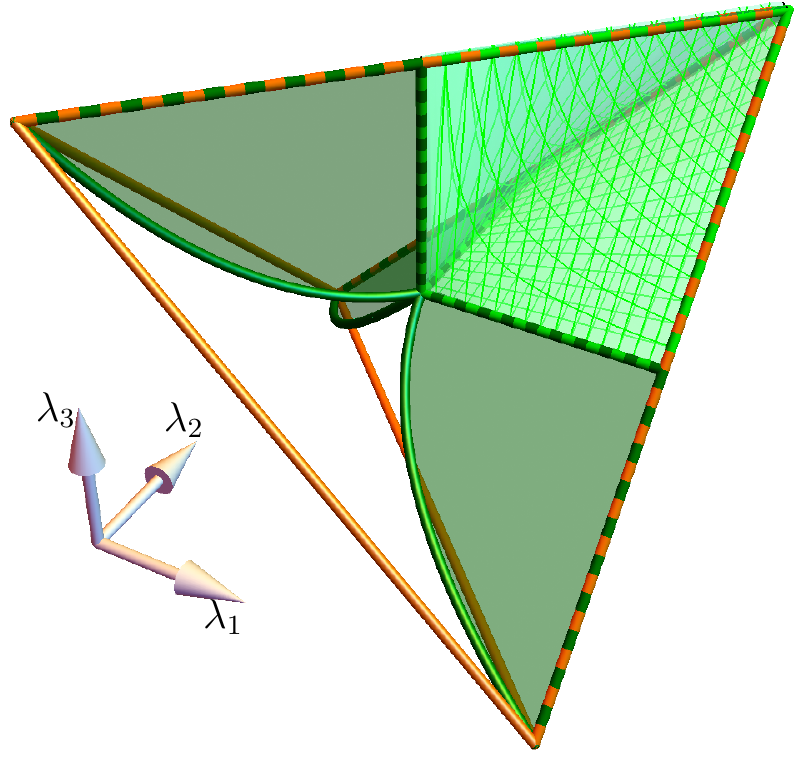}
\caption{Tetrahedron of Pauli channels with the set of L-divisible channels
(or equivalently infinitely divisible, see Theorem~\ref{thm:pauli_infinity}) shown in green, see  equations (\ref{eq:hermpres}) and (\ref{eq:ccp_degenerate}). The solid set corresponds to channels with
positive eigenvalues, and the 2D sets correspond to the negative eigenvalue
case. The point where the four sets meet corresponds to the \textit{total
depolarizing} channel. Notice that this set {\it does not} have the symmetries
of the tetrahedron. \label{fig:markov}}
\end{figure} 
\begin{figure*} 
\centering
\includegraphics[]{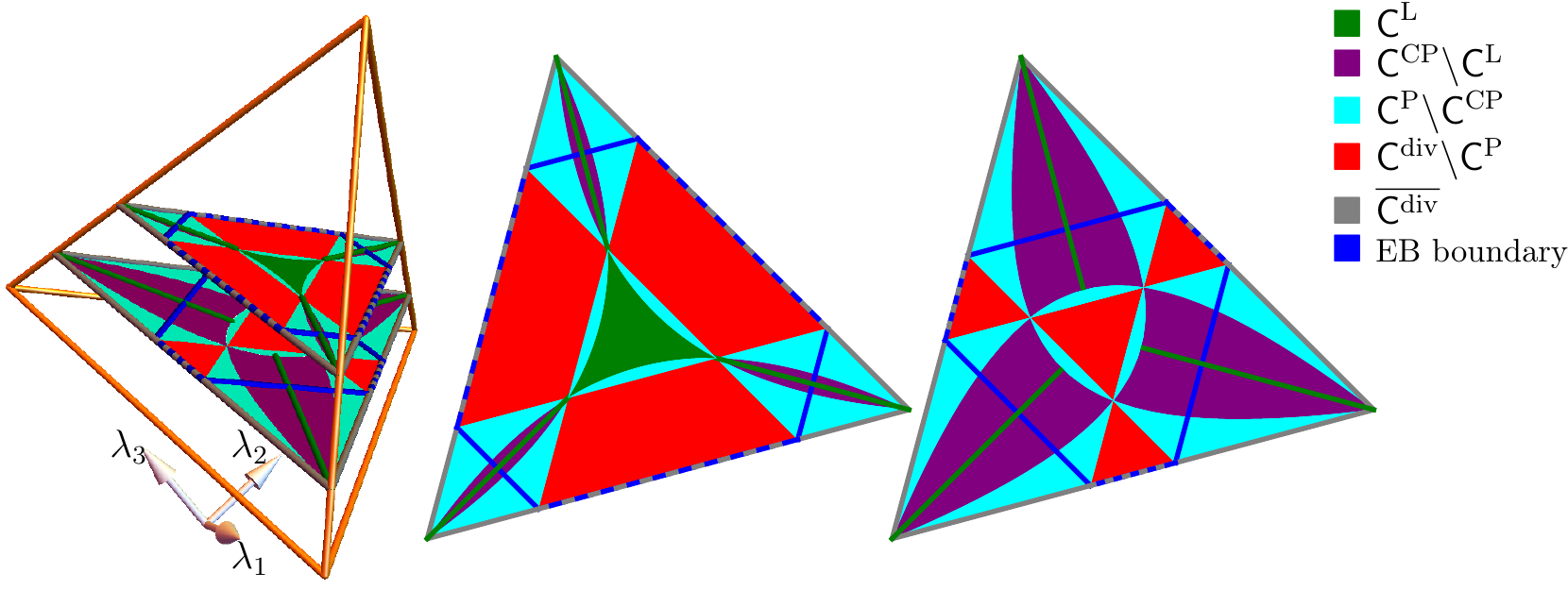}
\caption{We show two slices of the unitary tetrahedron (figure in the left) determined by 
$\sum_i \lambda_i=0.4$ (shown in the center) and $\sum_i \lambda_i=-0.4$ (shown in the right). The non-convexity of the
divisibility sets  can be seen, including the set of indivisible channels. The convexity
of sets ${\sf C}$ and entanglement breaking channels can also be noticed in the slices. A central
feature is that the set $\Div{}\setminus\pDiv{}$ is always inside the
octahedron of entanglement breaking channels.  \label{fig:cut1}}
\end{figure*} 
\subsection{Non-unital qubit channels} 
\label{subsec:generalqubitchannels}
Similar to unital channels, using
theorem~\ref{thm:divisibility_using_orthogonal_form} we are able to
characterize \Div{}, \pDiv{} and \cpDiv{} by studying special orthogonal normal
forms of non-unital channels. They are characterized by $\vec \lambda$ and
$\vec \tau$, see~\eref{eq:orthogonalform}. Thus, we can study if a channel is
\Div{} by computing the rank of its Choi matrix. For this case algebraic
equations are in general fourth order polynomials. In fact, 
in Ref.~\cite{Lukasz} a condition in terms of the eigenvalues and $\vec \tau$
is given.  For special cases, however, we can
obtain compact expressions, see~\fref{fig:cutnonunital2}. The characterization
of \pDiv{} is given by~\eref{eq:pdiv_qubits} (note that it only depends on
$\vec \lambda$), and \cpDiv{} is tested, for full Kraus rank non-unital channels,
using~\eref{eq:cpfullkraus_qubits}, see Ref.~\cite{Verstraete2001} for the
calculation of the $s_i$'s. For the characterization of \LDiv{} we use
the results developed at the end of the last section, see Eqs.~(\ref{eq:ccp_degenerate})-(\ref{eq:ccp_complex}).

We can plot illustrative pictures even though the whole space of qubit
channels has $12$ parameters. This can be done using orthogonal normal forms
and fixing $\vec \tau$, exactly in the same way as the unital case.
Recall that unitaries only modify \LDiv{}, leaving the shape of other sets
unchanged. CPTP channels are represented as a volume inside the
tetrahedron presented in \fref{fig:tetra}, see \fref{fig:cutnonunital2}.  In
the later figure we show a slice corresponding to $\vec \tau=\left(1/2, 0,0
\right)^\text{T}$. Indeed, it has the same structure of the slices for the
unital case, but deformed, see \fref{fig:cut1}. 
A difference with respect to the unital case is that L-divisible
channels with negative eigenvalues (up to unitaries) are not completely inside
CP-divisible channels. A part of them are inside the \pDiv{}
channels.

A central feature of Figs. \ref{fig:cut1} and \ref{fig:cutnonunital2} is that
the set $\Div{}\setminus\pDiv{}$ is inside the convex slice of the set of
entanglement breaking channels (deformed octahedron). Indeed, we can proof the
following theorem. 
\begin{theorem}[Entanglement breaking channels and divisibility]
  Consider a qubit channel $\mcE$.
  If $\det\hat\mcE<0$, then $\mcE$ is entanglement breaking,
  i.e. all qubit channels outside $\pDiv{}$ are entanglement breaking.
\label{thm:EB}
\end{theorem}
\begin{proof}
Consider the \Jami{} state of a channel $\mcE$ written
in the factorized Pauli operator basis
$\tau_{\mcE}=\frac{1}{4}\sum^{3}_{jk}R_{jk}\sigma_j\otimes \sigma_k$~\cite{Verstraete2001},
and let $\hat \mcE$ be its representation in the Pauli
operator basis. Then the matrix identity $R=\hat\mcE\Phi_{\text{T}}$
with $\Phi_{\text{T}}=\diag \left( 1,1,-1,1 \right)$ holds.
Since $\det \hat \mcE <0$ it follows that $\det R=-\det\hat{\mcE}>0$.
Using the aforementioned Lorentz normal decomposition
$R=L^{\text{T}}_1 \tilde R L_2$ with $\det L_{1,2}>0$, and $\tilde R$
diagonal for $\mcE$ with full Kraus rank, see Ref.~\cite{Verstraete2001}. The transformations
$L_{1,2}$ correspond to \textit{one-way stochastic local operations
  and classical communications} (1wSLOCC) of $\tau_{\mcE}$, thus,
$\tilde R$ corresponds to an unnormalized two-qubit state with $\det \tilde R>0$.
The channel corresponding to $\tilde R$ (in the Pauli basis)
is $\hat {\mathcal G} = \tilde R \Phi_{\text{T}}/\tilde
R_{00}$. Since the latter is diagonal, then ${\mathcal G}$
is a Pauli channel with $\det\hat {\mathcal G}<0$. A Pauli channel
has a negative determinant if either all $\lambda_j$ are negative,
or exactly one of them is negative. In Ref.~\cite{Ziman2005} it has been
shown that the set of channels with $\lambda_j<0 \ \ \forall j$ are
entanglement breaking channels. Now, using the symmetries of the
tetrahedron, one can generate all channels with negative determinant
by concatenating this set with the Pauli rotations. Therefore
every Pauli channel with negative determinant is entanglement breaking, thus
$\tau_{\mathcal G}$ is separable. Given that 1wSLOCC operations
can not create entanglement~\cite{Horodecki}, we have that $\tau_\mcE$
is separable too. Therefore $\mcE$ is entanglement breaking.

The case when $\tilde R$ is non-diagonal corresponds to Kraus deficient channels (the matrix rank of \eref{eq:state_normal_form_singular} is at most $3$). This case can be analyzed as follows. 
Since the neighborhood of any Kraus deficient channel with negative determinant contains full Kraus rank channels, by continuity of the determinant such channels have negative determinant too. The last ones are entanglement breaking since full Kraus rank channels have diagonal Lorentz normal form. Therefore, by continuity of the concurrence~\cite{Ziman2005}, Kraus deficient channels with negative determinant are entanglement breaking.
\end{proof}

\begin{figure} 
\centering
\includegraphics{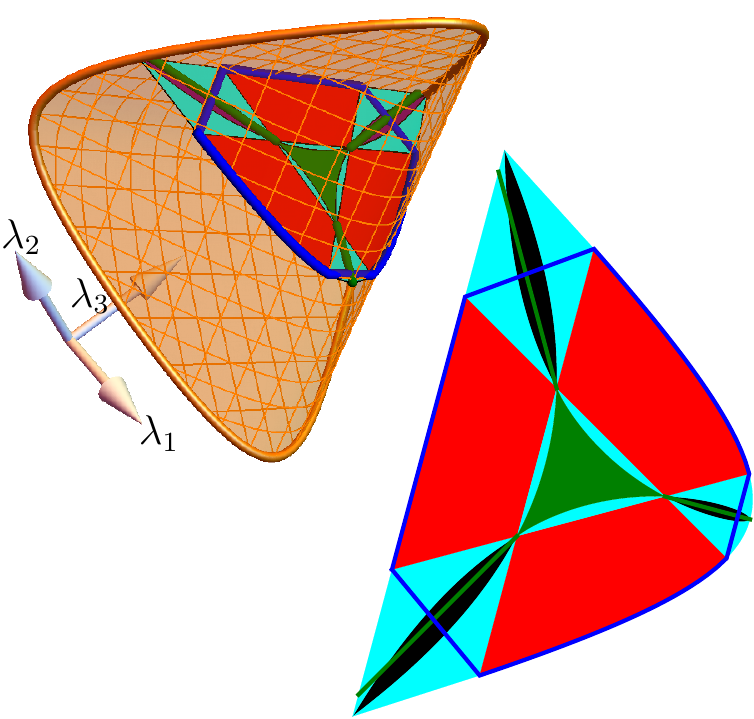}
\caption{(left) Set of non-unital unital channels up to unitaries, defined by
$\vec \tau=(1/2,0,0)$, see \eref{eq:orthogonalform}. This set lies inside the
tetrahedron. For this particular case the CP
conditions reduce to the two
inequalities $2 \pm 2\lambda_1 \ge \sqrt{1+4(\lambda_2 \pm \lambda_3)^2}$.
A cut corresponding to $\sum_i \lambda_i = 0.3$ is presented inside and in the right, see \fref{fig:cut1} for the color coding. The structure of divisibility sets presented here has basically the same structure as for the unital case except for \LDiv{}. A part of the channels with negative eigenvalues belonging to \LDiv{} lies outside $\cpDiv\setminus \LDiv$, see green lines. As for the unital case a central feature is that the channels in $\Div{}\setminus \pDiv{}$ are entanglement breaking channels. Channels in the boundary are not characterized due to the restricted character of Theorem~\ref{thm:Lorentz}.
 \label{fig:cutnonunital2}}
\end{figure} 
\section{Divisibility transitions and examples with dynamical process} 
\label{sec:jumps}
The aim of this section is to use illustrative examples of quantum dynamical
processes to show transitions between divisibility types of the instantaneous
channels. From the slices shown above (see figures \ref{fig:cut1} and
\ref{fig:cutnonunital2}) it can be noticed that every transition between the
studied divisibility types is permitted. This is due to the existence of common
borders between all combinations of divisibility sets; we can think of any
continuous line inside the tetrahedron~\cite{filippov2017} as describing some
quantum dynamical map. 

We analyze two examples, the first is an implementation of the approximate NOT gate, $\mcA_\text{NOT}$ throughout a specific collision model~\cite{Rybar2012}. 
The second is the well known setting of a two-level atom interacting with a quantized mode of an optical cavity \cite{Haroche06}. We define a simple function that assigns a particular value to a channel $\mcE_t$ according to divisibility hierarchy, i.e.
\begin{equation}
\DivFunc[\mcE]=\left\{
\begin{array}{cl}
1&\ {\rm if}\ \mcE\in{\sf C}^{\rm L}\,,\\
2/3&\ {\rm if}\ \mcE\in{\sf C}^{\rm CP}\setminus{\sf C}^{\rm L}\,,\\
1/3&\ {\rm if}\ \mcE\in{\sf C}^{\rm P}\setminus{\sf C}^{\rm CP}\,,\\
0&\ {\rm if}\ \mcE\in{\sf C}\setminus{\sf C}^{\rm P}\,.
\end{array}
\right.
\label{eq:delta_function}
\end{equation}
A similar function can be defined
to study the transition to/from the set of entanglement breaking channels, \ie{}
\begin{equation}
\chi[\mcE]=\left\{
\begin{array}{cl}
1&\ {\rm if}\ \mcE\text{ is entanglement breaking}\,,\\
0&\ {\rm if}\ \mcE\text{ if not}.
\end{array}
\right.
\label{eq:chi_function}
\end{equation}

The quantum NOT gate is defined as $\text{NOT}:\rho \mapsto \one-\rho$, \ie{}
it maps pure qubit states to its orthogonal state. Although this map transforms
the Bloch sphere into itself it is not a CPTP map, and the closest CPTP map
is $\mcA_{\rm NOT}:\rho \mapsto
(2\one-\rho)/3$. This is a rank-three qubit unital channel, thus,
it is indivisible \cite{cirac}. Moreover, $\det \mcA_{\rm NOT}=-1/27$ implies
that this channel is not achievable by a P-divisible dynamical map.  It is worth
noting that $\mcA_\text{NOT}$ belongs to \Ind{}.

A specific collision model was designed in Ref.~\cite{Rybar2012} simulating
stroboscopically a quantum dynamical map that implements the quantum NOT
gate $\mcA_{\rm NOT}$ in finite time.
The model reads
\begin{equation}
\mcE_t(\varrho)=\cos^2( t)\varrho+\sin^2( t)\mcA_{\rm NOT}(\varrho)
+\frac{1}{2}\sin(2 t)\mcF(\varrho)\,,
\label{eq:collision:model}
\end{equation}
where $\mcF(\varrho)=i\frac{1}{3}\sum_j [\sigma_j,\varrho]$.
This quantum dynamical map achieves the desired gate
$\mcA_{\rm NOT}$ at $t=\pi/2$.

Let us stress that this dynamical map is unital, i.e. $\mcE_t(\one)=\one$ for
all $t$, thus, its orthogonal normal form can be illustrated inside the
tetrahedron of Pauli channels, see  \fref{fig:tray}. In
\fref{fig:evolnot} we plot $\DivFunc{}[\mcE_t]$, $\chi[\mcE_t]$ and the
value of the $\det\mcE_t$.  We see the transitions ${\sf C}^{\rm L}
\rightarrow{\sf C}^{\rm P}\setminus{\sf C}^{\rm CP} \rightarrow{\sf C}^{\rm
div}\setminus{\sf C}^{\rm P} \rightarrow \Ind{}$ and back. Notice that
in both plots the trajectory never goes through the $\cpDiv{}\setminus \LDiv{}$
region. This means that when the parametrized channels up to rotations belong
to \LDiv{}, so do the original ones. The transition
between P-divisible and divisible channels, i.e. \pDiv{}$\setminus$\cpDiv{} and
\Div{}$\setminus$\pDiv{}, occurs at the discontinuity in the yellow curve in 
\fref{fig:tray}. Let us note that this discontinuity 
only occurs in the space of $\vec \lambda$; it is a consequence of the
orthogonal normal decomposition, see \eref{eq:orthogonalform}. The complete
channel is continuous in the full convex space of qubit CPTP maps.
 The transition from
$\pDiv{}\setminus \Div{}$ and back occurs at times $\pi/3$ and $2 \pi/3$.  It
can also be noted that the transition to entanglement breaking channels 
occurs shortly before the channel enters in the $\Div{} \setminus \pDiv{}$
region; likewise, the channel stops being entanglement breaking shortly after
it leaves the $\Div{} \setminus \pDiv{}$ region.

\begin{figure} 
\centering
\includegraphics[]{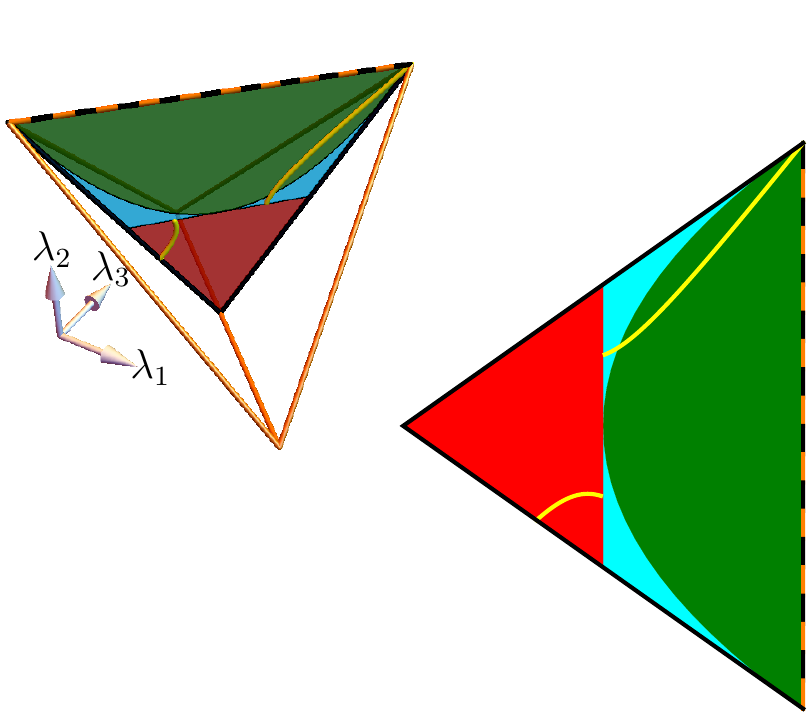}
\caption{(top left) Tetrahedron of Pauli channels with the trajectory, up to
rotations, of the quantum dynamical map~\eref{eq:collision:model} leading to
the $\mcA_\text{NOT}$ gate, as a yellow curve. (right) Cut along
the plane that contains the trajectory; there one can see the different regions
where the channel passes. For this
case, the characterization of the \LDiv{} of the channels induced gives the
same conclusions as for the corresponding Pauli channel, see
\eref{eq:orthogonalform}.
The discontinuity in the trajectory is due to the reduced representation of the
dynamical map, see \eref{eq:collision:model}; the trajectory is continuous
in the space of channels. 
See \fref{fig:cut1} for the color coding.
\label{fig:tray}}.
\end{figure} 

\begin{figure} 
\centering
\includegraphics[width=\columnwidth]{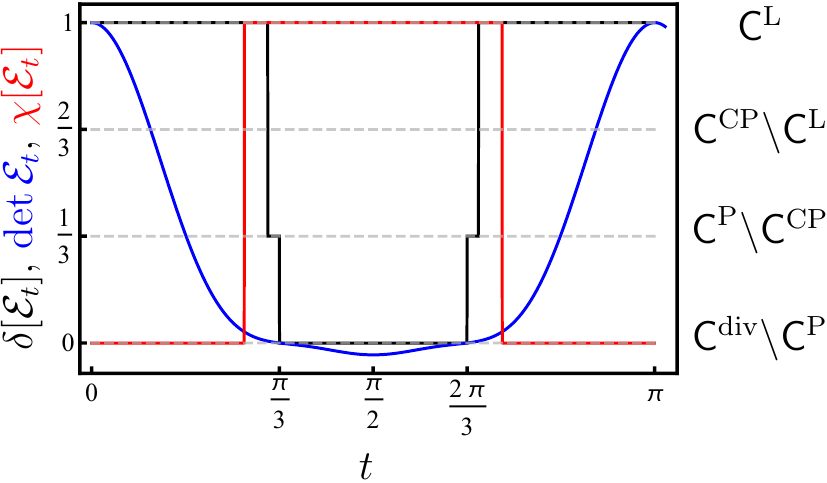}
\caption{
Evolution of divisibility, determinant, and entanglement breaking properties
of the map induced by \eref{eq:collision:model},
see \eref{eq:delta_function} and \eref{eq:chi_function}.
Notice that the channel $\mcA_\text{NOT}$, implemented at $t=\pi/2$, has minimum determinant. The horizontal gray dashed lines show the image
of the function $\delta$, with the divisibility types in the right side. It can
be seen that the dynamical map explores the divisibility sets as ${\sf C}^{\rm
L} \rightarrow{\sf C}^{\rm P}\setminus{\sf C}^{\rm CP} \rightarrow{\sf C}^{\rm
div}\setminus{\sf C}^{\rm P} \rightarrow \Ind{}$ and back. The channels are
entanglement breaking in the expected region.\label{fig:evolnot}}
\end{figure} 

Consider now the dynamical map induced by a two-level atom interacting with a mode of a
boson field. This model serves as a workhorse to explore a great variety of
phenomena in quantum optics~\cite{0953-4075-46-22-220201}. Using
the well known \textit{rotating wave approximation} one arrives to the
Jaynes-Cummings model~\cite{jaynescummings}, whose Hamiltonian is
\begin{equation}
H = \frac{\omega_a}{2}\sigma_z
    +\omega_f \left( a^{\dagger}a+\frac{1}{2} \right)
    + g\left( \sigma_- a^ {\dagger}+ \sigma_+ a \right).
\label{eq:Ham_qb}
\end{equation}
By initializing the environment in a coherent state $\ket{\alpha}$,
one gets the familiar \textit{collapse and revival} setting. Considering a particular set of parameters shown in~\fref{fig:revival1}, we
constructed the channels parametrized by time numerically, and studied their
divisibility and entanglement breaking properties. In the same figure we plot
functions $\DivFunc{}[\mcE_t]$ and $\chi[\mcE_t]$, together with the probability of
finding the atom in its excited state $p_e(t)$, to study and compare the
divisibility properties with the features of the collapses and revivals. The
probability $p_e(t)$ is calculated choosing the ground state of the free
Hamiltonian  $\omega_a /2 \sigma_z$ of the qubit, and it is given
by~\cite{klimovbook}:
\begin{equation}
p_e(t)=\frac{\langle\sigma_z(t)\rangle+1}{2},
\end{equation}
where 
\begin{equation*}
\langle \sigma_z(t)\rangle=-\sum_{n=0}^{\infty}P_n \left( \frac{\Delta^2}{4 \Omega_n^2}+\left(1-\frac{\Delta^2}{4 \Omega_n^2}\right)\cos \left(2 \Omega_n t \right)\right),
\end{equation*}
with $P_n=e^{-|\alpha|^2}|\alpha|^{2n}/n!$, $\Omega_n=\sqrt{\Delta^2/4+g^2n}$ and $\Delta=\omega_f-\omega_a$ the detuning.
\begin{figure} 
\centering
\includegraphics[width=\columnwidth]{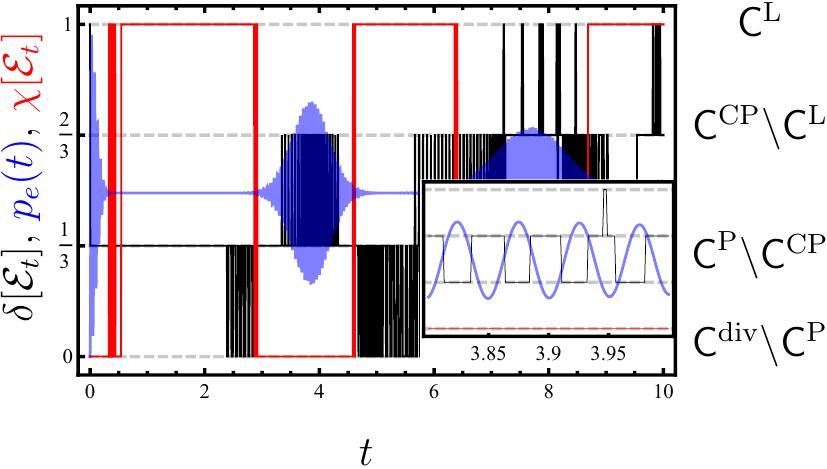}
%
\caption{Black and red curves show functions $\delta$ and $\chi$ of the
channels induced by the Jaynes-Cummings model over a two-level system,
see~\eref{eq:Ham_qb} with the environment initialized in a coherent state
$\ket{\alpha}$. The blue curve shows the probability of finding the
two-level atom in its excited state, $p_e(t)$.
The figure shows that the fast oscillations in $\delta$ occur
roughly at the same frequency as the ones of $p_e(t)$, see the inset. Notice
that there are fast transitions between $\pDiv{}\setminus\cpDiv{}$ and
$\cpDiv{}\setminus\LDiv{}$ occurring in the region of revivals, with a few
transitions between $\cpDiv{}\setminus\pDiv{}$ and $\LDiv{}$ in the second
revival. The function $\chi$ shows that during revivals channels are not
entanglement breaking, but we find that channels belonging to
$\Div{}\setminus\pDiv{}$ are always entanglement breaking, in agreement with
theorem~\ref{thm:EB}. The particular chosen set of parameters are
$\alpha=6$, $g=10$, $\omega_a=5$, and $\omega_f=20$.  \label{fig:revival1}}
\end{figure} 

The divisibility indicator function $\delta$ exhibits
an oscillating behavior, roughly at the same frequency of $p_e(t)$, see inset
in \fref{fig:revival1}.  The figure shows fast periodic transitions
between $\pDiv{}\setminus\cpDiv{}$ and $\cpDiv{}\setminus\LDiv{}$ occurring in
the region of revivals. There are also few transitions among
$\cpDiv{}\setminus\pDiv{}$ and $\LDiv{}$ in the second revival. Respect to the
entanglement breaking and the function $\chi$, there are no fast transitions in the former, and during revivals, channels are not entanglement breaking. We also observe that channels belonging to $\Div{}\setminus \pDiv{}$ are entanglement breaking, supporting theorem~\ref{thm:EB} for the non-unital case.
\section{Conclusions} 
\label{sec:conclusions}
We studied the relations between different types of divisibility of
time-discrete and time-continuous quantum processes, i.e. channels and
dynamical maps, respectively. In particular, we investigated
classes of channels by means of their achievability by dynamical maps
of different divisibility types, and also the divisibility of channels
occurring during the time evolutions. Apart from investigating the
relations between these concepts in general, we provided a detailed
analysis for the case of qubit channels. 

We implemented the known conditions to decide \LDiv{} for the general
diagonalizable case, and a discussion of the parametric space of Lindblad
generators was given (clarifying one of the results of the paper
\cite{Wolf2008}). For unital qubit channels it was shown that every
infinitesimal divisible map can be written as a concatenation of one \LDiv{}
channel and two unitary conjugations. For the particular case of Pauli channels
case, we have shown that the sets of infinitely divisible and L-divisible
channels coincide. We made an interesting observation, connecting the concept
of divisibility with the quantum information paradigm of entanglement breaking
channels. We found that divisible but not infinitesimal divisible qubit
channels, in PTP maps, are necessarily entanglement breaking. We also noted
that the intersection of indivisible and P-divisible channels is not empty.
This allows us to implement indivisible channels with infinitesimal PTP maps.
Finally, we questioned the existence of dynamical transitions between different
classes of divisibility channels. We argued that all the transitions are, in
principle, possible, and exploited two simple models of dynamical maps to
demonstrate these transitions. They clearly illustrate how the channels
evolutions change from being implementable by markovian dynamical maps to
non-markovian, and vice versa.

There are several directions how to proceed further in investigation of
divisibility of channels and dynamical maps. Apart from extension of this
analysis to larger-dimensional systems, a plethora of interesting questions are
related to design of efficient verification procedures of the divisibility
classes for channels and dynamical maps. 
In this paper we question the
divisibility features of snapshots of the evolution, however, it
might be of interest to understand when the time intervals of dynamical maps
implemented by non-markovian evolutions, can be simulated by markovian dynamical maps. 
 Also the area of channel divisibility contains several open
structural questions, e.g. the existence of at most $n$-divisible channels.
\section*{Acknowledgements} 
We acknowledge Thomas Gorin and Tom\'a\v{s} Ryb\'ar for useful discussions, as
well PAEP and RedTC for financial support. Support by projects CONACyT 285754,
UNAM-PAPIIT IG100518, IN-107414, APVV-14-0878 (QETWORK) is acknowledged. CP
acknowledges support by PASPA program from DGAPA-UNAM. MZ acknowledges the support of VEGA
2/0173/17 (MAXAP), GA\v CR project no.
GA16-22211S and MUNI/G/1211/2017 (GRUPIK).
\appendix
\section{On Lorentz normal forms of Choi-Jamiolkowski state} 
\label{sec:appendix}
In this appendix we compute the Lorentz normal decomposition of a channel for
which one gets $b\neq 0$, supporting our observation that Lorentz normal
decomposition does not take \Jami{} states to something proportional to a \Jami{}
state. Consider the following Kraus rank three channel and its
$R_{\mcE}$ matrix, both written in the Pauli basis:
\begin{equation}
\hat \mcE=\left(
\begin{array}{cccc}
 1 & 0 & 0 & 0 \\
 0 & -\frac{1}{3} & 0 & 0 \\
 0 & 0 & -\frac{1}{3} & 0 \\
 \frac{2}{3} & 0 & 0 & \frac{1}{3} \\
\end{array}
\right),
\end{equation}
and 
\begin{equation}
R_{\mcE}=\left(
\begin{array}{cccc}
 1 & 0 & 0 & 0 \\
 0 & -\frac{1}{3} & 0 & 0 \\
 0 & 0 & \frac{1}{3} & 0 \\
 \frac{2}{3} & 0 & 0 & \frac{1}{3} \\
\end{array}
\right).
\end{equation}
Using the algorithm introduced in Ref.~\cite{Verstraete2001} to 
calculate $R_{\mcE}$'s  Lorentz
decomposition into orthochronous proper Lorentz transformations
we obtain
\begin{align}
L_1 &= 
\frac{1}{\gamma_1}
\begin{pmatrix}
 4 & 0 & 0 & 1 \\
 0 & -\gamma_1 & 0 & 0 \\
 0 & 0 & -\gamma_1 & 0 \\
 1 & 0 & 0 & 4
\end{pmatrix}, 
\end{align}
\begin{align*}
L_2 & = 
\frac{1}{\gamma_2}
\begin{pmatrix}
 89+9\sqrt{97} & 0 & 0 & -8 \\
 0 & -\gamma_2 & 0 & 0 \\
 0 & 0 & -\gamma_2 & 0 \\
 -8 & 0 & 0 & 89+9\sqrt{97}
\end{pmatrix},
\end{align*}
and
\begin{align*}
\Sigma_\mcE&=
\frac{1}{\gamma_3}
\begin{pmatrix}
 \sqrt{11+\frac{109}{\sqrt{97}} } & 0 & 0 & -\frac{\sqrt{97}+1}{\sqrt{ 89 \sqrt{97}+873}} \\
 0 & -\frac{\gamma_3}{3} & 0 & 0 \\
 0 & 0 & \frac{\gamma_3}{3} & 0 \\
 \sqrt{1+\frac{49}{\sqrt{97} }} & 0 & 0 & \sqrt{-1+\frac{49}{\sqrt{97} }} \\
\end{pmatrix}
\end{align*}
with $\gamma_1=\sqrt{15}$, $\gamma_2=3\sqrt{178 \sqrt{97}+1746}$, and 
$\gamma_3=\sqrt{30}$.
Although the central matrix $\Sigma_\mcE$ is not exactly of the form
\eref{eq:state_normal_form_singular}, it is equivalent. To see this notice
that the derivation of the theorem 2 in~\cite{Verstraete2001} considers
only decompositions into proper orthochronous Lorentz transformations. But 
to obtain the desired form, the authors change signs until they get
\eref{eq:state_normal_form_singular}; this cannot be done without changing
Lorentz transformations. If we relax the condition over $L_{1,2}$ of being proper
and orthochronous, we can bring $\Sigma_\mcE$ to the desired form by
conjugating $\Sigma_\mcE$ with $G=\diag\left(1,1,1,-1 \right)$:
\begin{multline*}
G^{-1} \Sigma_\mcE G= \\
\frac{1}{\gamma_3}
\begin{pmatrix}
 \sqrt{11+\frac{109}{\sqrt{97} }} & 0 & 0 & \frac{\sqrt{97}+1}{\sqrt{ 89 \sqrt{97}+873}} \\
 0 & -\frac{\gamma_3}{3} & 0 & 0 \\
 0 & 0 & \frac{\gamma_3}{3} & 0 \\
- \sqrt{1+\frac{49}{\sqrt{97} }} & 0 & 0 & \sqrt{-1+\frac{49}{\sqrt{97} }} \\
\end{pmatrix}.
\end{multline*}
In both cases (taking $\Sigma_\mcE$ or $G^{-1} \Sigma_\mcE G$ as the normal form of $R_\mcE$), the corresponding channel is not proportional to a trace-preserving
one since $b\neq 0$, see \eref{eq:state_normal_form_singular}.
This completes the counterexample.
\bibliographystyle{plainnat}
\bibliography{labibliografia}

\end{document}